\documentclass[11pt]{article}
\usepackage[utf8]{inputenc}
\usepackage{amsfonts}
\usepackage{amssymb}
\usepackage{fullpage,tikz}

\usepackage{mathtools}
\usepackage{amsmath,amsthm, amssymb,amscd,amstext,amsfonts}
\usepackage[utf8]{inputenc}
\usepackage{mathtools}
\usepackage{mathrsfs}
\usepackage{thmtools}
\usepackage{latexsym}
\usepackage{verbatim}
\usepackage{framed}
\usepackage{graphicx}
\usepackage{stmaryrd}
\usepackage{enumitem}
\usepackage{fullpage}
\usepackage{bm}
\usepackage{url}
\usepackage{physics}
\usepackage{dsfont}
\usepackage{todonotes}
 
\usepackage{mathtools}

\usepackage{color}
\usepackage{hyperref}
\hypersetup{
    colorlinks=true,
    linkcolor=blue,
    citecolor=blue
    }

\usepackage[nameinlink, noabbrev, capitalize]{cleveref}

\theoremstyle{plain}

\newtheorem{thm}{Theorem}[section]
\newtheorem{theorem}[thm]{Theorem}

\newtheorem{problem}[thm]{Problem}

\newtheorem*{claim*}{Claim}

\theoremstyle{remark}

\newcommand{\set}[1]{\{#1\}}
 
\newcommand{\R}{\mathbb{R}}

\newcommand{\cL}{\mathcal L}

\newcommand{\cP}{\mathcal P}

\newcommand{\defeq}{\coloneqq}

\def\1{\mathbf{1}} 
\def\0{\mathbf{0}}

\DeclareMathOperator{\AC}{AC}

\DeclareMathOperator{\Cov}{Cov}

\title{On depth-3 circuits and covering number: an explicit counter-example}

\author{Lianna Hambardzumyan \thanks{School of Computer Science, McGill University. \texttt{lianna.hambardzumyan@mail.mcgill.ca}} 
\and Hamed Hatami \thanks{School of Computer Science, McGill University. \texttt{hatami@cs.mcgill.ca}. Supported by an NSERC grant.}
\and Ndiam\'e Ndiaye \thanks{Department of Mathematics and Statistics, McGill University.  \texttt{ndiame.ndiaye@mail.mcgill.ca}}}

\begin{document}

\maketitle

\begin{abstract}
We give a simple construction of $n\times n$  Boolean matrices with  $\Omega(n^{4/3})$ zero entries that are free of $2 \times 2$ all-zero submatrices and have covering number   $O(\log^4(n))$. This construction provides an \emph{explicit} counterexample to a conjecture of Pudl\'{a}k, R\"{o}dl and Savick\'{y}~\cite{MR948752} and  Research Problems 1.33, 4.9, 11.17  of Jukna~\cite{MR2895965}.  These conjectures were previously  refuted by Katz~\cite{MR3001038} using a probabilistic construction. 
\end{abstract}

\section{Introduction}

The \emph{covering number} of a Boolean matrix $A$, denoted by $\Cov(A)$, is the smallest number of all-one submatrices covering all the one-entries of $A$.   
Covering number is a well-studied notion in complexity theory as its  logarithm is the nondeterministic communication complexity of the matrix~\cite[Definition 2.3]{MR1426129}. However, 
our study of this notion  is motivated by its connections to circuit complexity.   A fruitful line of work~\cite{10.1109/SFCS.1981.35,DBLP:journals/apal/Ajtai83,10.1109/SFCS.1985.49} in the 80's lead to  H\r{a}stad's~\cite{10.5555/27669} lower bound of $2^{\Omega(k^{\frac{1}{d-1}})}$ on the size of  any depth-$d$  $\AC$-circuit  computing the parity function on $k$ bits. While these bounds are optimal for parity, a counting argument shows that (regardless of the depth of the circuit) most functions require circuits of size $2^{\Omega(k)}$. However, to this day, H\r{a}stad's bound remains the strongest explicit known lower bound against small-depth circuits for any
function, even in the case of $d = 3$. 

The special case of  depth-3 in particular has received significant attention as one of the simplest  restricted models where our understanding is lacking~\cite{MR948752,DBLP:journals/cc/HastadJP95,Lecomte}. In~\cite[page 523]{MR948752},  
Pudl\'{a}k, R\"{o}dl and Savick\'{y} proposed a graph theoretical approach towards obtaining stronger explicit lower bounds on size  of these circuits. Later, Jukna refined this approach further and showed that  a positive answer to the following problem  would  imply  the incredibly strong explicit lower bound of  $2^{\Omega(k)}$ against depth-$3$ circuits (See~\cite[Problem 11.17]{MR2895965} and the discussion that follows it).   

\begin{problem}\label{main-conj}
Does every $n \times n$ Boolean matrix $A$ with at least $dn$ zero-entries and no $2 \times 2$ all-zero submatrix satisfy   $\Cov(A) = d^{\Omega(1)}$?  
\end{problem}

Unfortunately,~\cref{main-conj} has been resolved in the negative by Katz~\cite{MR3001038} using a probabilistic argument. The purpose of this note is to provide a simple \emph{explicit}  counter-example for \cref{main-conj}.  

\section{Main result} 

Our main result  refutes~\cref{main-conj} with $d=\Omega(n^{1/3})$.  

\begin{theorem}\label{main-theorem}
There exist $n \times n$ Boolean matrices $A$ with $\Omega(n^{4/3})$ zero-entries and no  $2 \times 2$ all-zero submatrices that satisfy $$\Cov(A)=O(\log^4(n)).$$
\end{theorem}
\begin{proof}[Proof of \cref{main-theorem}]

Let $m$ be a positive integer and define the sets  
$ \cP \defeq \cL  \defeq [m] \times [2m^2]$.  We think of the elements $\ell=(\ell_1,\ell_2) \in \cL$ as  lines $y=\ell_1x + \ell_2$ in $\R^2$, and the elements $p=(p_1,p_2) \in \cP$ as points in $\R^2$. Define the Boolean matrix $M_{\cP,\cL}$ according to point line incidences: 
$$
M_{p,\ell}=
\begin{cases}
0 & p \in \ell\\
1 & p \not\in \ell
\end{cases}, 
$$
or equivalently 
$$M_{p,\ell}=\begin{cases} 0 & \ell_1 p_1+\ell_2=p_2 \\ 1 & \ell_1p_1+\ell_2 \neq p_2 \end{cases},
$$
for every $p=(p_1,p_2)\in \cP$ and $\ell=(\ell_1,\ell_2)\in \cL$. Let $n=2m^3$, and observe that $M$ is an $n \times n$ matrix that satisfies the conditions of~\cref{main-theorem}:
\begin{enumerate}[label=(\roman*)]
    \item $M$ has no $2 \times 2$ all-zero submatrices since there is at most one line in $\cL$ that passes through any two distinct points in $\cP$. 
    \item $M$ has at least $m^4$ zero-entries: For any fixed $p_1,\ell_1\in [m]$, and $\ell_2\in [m^2]$, the point $(p_1,p_2) \defeq (p_1,\ell_1 p_1+\ell_2) \in \cP$ is on the line $(\ell_1,\ell_2) \in \cL$. 
\end{enumerate}

We begin constructing the cover for $M$. Let $q$ be a prime number and take $(a,b,c)\in[q]^3$. Define the sets
$$\cP^{q}_{a,b,c}=\set{(p_1,p_2)\in \cP \ : \  (p_1 \equiv a \mod q) \ \text{ and }\ (p_2\not\equiv ab+c \mod q)},$$
$$\cL^{q}_{a,b,c}=\set{(\ell_1,\ell_2)\in \cL \ : \ (\ell_1 \equiv b \mod q) \ \text{ and } \ (\ell_2 \equiv c\mod q)},$$
and 
$$R^{q}_{a,b,c} = \cP^{q}_{a,b,c}\times \cL^{q}_{a,b,c}.$$
It is clear that for any choice of $q$ and for any $(a,b,c) \in [q]^3$, the set $R^q_{a,b,c}$ is a $1$-monochromatic rectangle of $M$ because any $[(p_1,p_2),(\ell_1,\ell_2)] \in R^q_{a,b,c}$  satisfies $\ell_1p_1+\ell_2 \neq p_2 \mod q$. Next, fix $k = \lceil \log_2 (2m^2) \rceil$ and let $Q_k = \{q : q < k \text{ and  } q \text{ is a prime number}\}$.
Define the set \[S_k=\bigcup_{q\in Q_k} \bigcup_{(a,b,c)\in [q]^3} R^{q}_{a,b,c}.\]

We claim that $S_k$ covers all the one-entries of $M$. If  $\big((p_1,p_2),(\ell_1,\ell_2)\big)$ is not covered by $S_k$, then   for all $q\in Q_k$, we have $p_2=(p_1 \ell_1+\ell_2)\mod q$. Consequently,  $p_2=(p_1\ell_1+\ell_2) \mod \prod_{q\in Q_k} q$. However, since both $p_2$ and $\ell_1 p_1+\ell_2$ are at most $2m^2 < \prod_{q\in Q_{k}} q$,   they can only be equal modulo $\prod_{q\in Q_k} q$ if they are equal. This implies that $M_{p,\ell}=0$, so all the 1's are covered. Thus, by the choice of $k$ and $S_k$, all the one-entries of $M$ can be covered using $O(\log^4 m)$ monochromatic rectangles.

\end{proof}

 Lastly, we note that even though the answer to~\cref{main-conj} is negative, the general approach of using covering number to obtain explicit lower bounds against depth-$3$ circuits could still be possible. For example, \cite[Research Problem 4.8]{MR2895965}  remains open.
\bibliographystyle{plain}
\bibliography{ref}

\end{document}